\documentclass{article}
\usepackage{arxiv}
\usepackage[utf8]{inputenc} % allow utf-8 input
\usepackage[T1]{fontenc}    % use 8-bit T1 fonts
\usepackage{hyperref}       % hyperlinks
\usepackage{url}            % simple URL typesetting
\usepackage{booktabs}       % professional-quality tables
\usepackage{amsfonts,amsmath}% blackboard math symbols
\usepackage{amsfonts}       % blackboard math symbols
\usepackage{amssymb, amsthm}
\usepackage{lmodern}
\usepackage{mathtools}
\usepackage[titletoc]{appendix}
\DeclarePairedDelimiter{\ceil}{\lceil}{\rceil}
\usepackage{float}
\usepackage{color}
\newtheorem{proposition}{Proposition}
\usepackage{physics}
\usepackage{nicefrac}       % compact symbols for 1/2, etc.
\usepackage{microtype}      % microtypography
\usepackage{lipsum}
\usepackage{graphicx}
\graphicspath{ {./images/} }
\usepackage{authblk}

\title{Optimal upper bound of entropic uncertainty relation for  mutually unbiased bases}

\author[1]{
 Bilal Canturk\thanks{bcanturk@sabanciuniv.edu} \textsuperscript{,}}
 \author[1]{Zafer Gedik}

 \affil[1]{ Faculty of Engineering and Natural Sciences, Sabanci
University \authorcr 
Tuzla, Istanbul 34956, Turkey} 

\begin{document}
\maketitle
\begin{abstract}
We have obtained the optimal upper bound of entropic uncertainty relation for $N$ Mutually Unbiased Bases (MUBs). We have used the methods of variational calculus for the states that can be  written in terms of $N$ MUBs. Our result is valid for any state when $N$ is $d+1$, where $d$ is the dimension of the related system. We provide a quantitative criterion for the extendibilty of MUBs. In addition, we have applied our result to the mutual information of $d+1$ observables conditioned with a classical memory.
\end{abstract}

% keywords can be removed
\keywords{Entropic uncertainty relation \and Mutually unbiased bases \and Mutually coherent state \and Extendibility of MUBs}
\pagenumbering{arabic}
\section{Introduction}
One of the fundamental tasks in the quantum information theory is how to extract the complete information of the density matrix of a system. For this purpose, an informationally complete set of measurement elements with rank-1 is performed so that is a maximally efficient measurement. Mutually Unbiased Bases (MUBs) \cite{Wootters1989} provide such a measurement. In addition to their importance in the vein of theoretical aspect \cite{Durt2010}, they have found room in diverse application areas such as quantum error correction \cite{Calderbank1997}, quantum cryptography \cite{Mafu2013}, entanglement detection \cite{Spengler2012} and quantum state tomography \cite{Ivanovic1981,Wootters1989}.
 
\par Uncertainty principle, however, puts a limit on obtaining information content of a quantum system. The observables corresponding to MUBs cannot be determined exactly; the more information about one of such observables is gained, the less information about the others is possible. This trade-off relation was first presented in terms of deviations ($\sigma_{i}$) of the observables by Heisenberg \cite{Heisenberg1927} and later, improved further \cite{Robertson1929,Schrodinger1930}. The expression of uncertainty principle in terms of deviations was formulated either as the product of the deviations or as the sum of them \cite{Mondal2017}. However, as firstly highlighted by Deutsch \cite{deutsch1983}, this formulation of uncertainty principle has some drawbacks; for example, lower bound of the uncertainty principle, $\sigma_{A}\left(\ket{\psi}\right) \sigma_{B}\left(\ket{\psi}\right)\geq \frac{1}{2}\abs{\mel{\psi}{\comm*{A}{B}}{\psi}}$, depends on the initial state, and thus, does not fix such that it can vanish for some choices of $\ket{\psi}$, which do not have to be simultaneous eigenfunctions of the observables $A$ and $B$. In addition, deviation-based uncertainty relations do not capture in general the physical content of the complementary aspect \cite{Dammeier2015}, and the spread of informational content \cite{birula2011}, of the observables. Expressing uncertainty in terms of  entropies of observables was first set forth as a question by Everett \cite{Everett1957}. It was answered affirmatively in Ref.\cite{Hirschman1957} such that the sum of entropies of position and momentum observables satisfies an inequality. This entropic uncertainty relation was proved and improved respectively in Refs.\cite{Beckner1975,Birula1975} for the observables of having a continuous spectrum. The lower bound of the inequality is achieved when the state of the system is a Gaussian wave-packet. The extension of entropic uncertainty relation to the observables in a finite dimensional Hilbert space was first presented in Ref.\cite{deutsch1983}, and improved later in Ref.\cite{Maassen1988}. We wish to highlight the importance of entropic uncertainty relation, which is that it does not have the aforementioned drawbacks of the uncertainty based on deviations.
Entropic uncertainty relation became a fundamental instrument in quantum information theory, especially for entanglement detection \cite{Prevedel2011,Zou2014}. It puts a lower bound on the summation of the entropies of two or more observables when they are measured. Formally, if $A_{n}$ and $A_{m}$ are two observables associated with a quantum system in Hilbert space $H^{d}$, with eigenvectors sets $\{\ket{nk}\}$ and $\{\ket{ml}\}$ respectively, then the summation of their entropies  has a lower bound \cite{Maassen1988},
\begin{equation*}
	H(A_{n})+ H(A_{m}) \geq -\ln(c), c= \max_{k,l}\left(\abs{\bra{nk}\ket{ml}} ^{2}\right),
\end{equation*}    
where $H(A_{n}):=-\sum_{k}p_{nk}\ln(p_{nk})$ is Shannon entropy of the observable $A_{n}$. This inequality was extended to the cases of when the system has some connection with its environment such as quantum memory \cite{Berta2010}. Beside Shannon entropy, other entropies, such as minimum entropy, collision entropy, Tsallis entropy, Rényi entropy, are also used according to their convenience to the relevant problem.  A review of entropic uncertainty relations and their applications can be found in Ref.\cite{Coles2017}. In addition to entropic uncertainty relation, upper bound of entropic uncertainty relation is another important concept which puts an upper bound on the summation of the entropies of two or more observables which, henceforth, we abbreviate as  entropic certainty relation. While entropic uncertainty relation quantifies the lack of information,  entropic certainty relation is related with the correlation between the observables. Entropic certainty relation for the observables set $\{A_{n}\}_{n=1}^{N}$ is defined as  $\sum_{n}H(A_{n})\leq f$. If such an upper bound is found then mutual information of the observables, which measures the correlation between the observables,
\begin{equation*}
\begin{split}
	I(A_{n} : Y):= H(A_{n})-H(A_{n} \mid Y)
\end{split}
\end{equation*} 
can also be bounded, where $Y$ is a classical (or quantum) memory given its access to the observer. In addition, entropic certainty relation can also be used for searching the existence of more than three MUBs especially when the dimension of the system is not a power of a prime number.  The extendibility of MUBs is one of the most important question in quantum information theory. We will return to this point in Sec.\ref{section2}.
\par  We obtain optimal entropic certainty relation of the measurements performed by $N$ MUBs for some density matrix. Our method is based on the variational calculus with some  conditions satisfied by the probability distributions.
\section{Optimal entropic certainty relation for  MUBs} \label{section2}
Two bases $\{\ket{nk}, k=1,2,\ldots,d \}$ and $\{\ket{ml}, l=1,2,\ldots,d\}$ of Hilbert space $H^{d}$, which may be considered as eigenvectors of two observables $A_{n}$ and $A_{m}$ respectively, are called mutually unbiased bases (MUBs) iff $\abs{\bra{nk}\ket{ml}}^{2}=1/d $, 
 for any $k,l$ and $n \neq m$. These observables, $A_{n}$ and $A_{m}$, are known as complementary, or mutually exclusive, observables. If there is $d+1$ MUBs then we reconstruct the density operator $\rho$ of a system by the aid of the outcomes of the measurement of the observables as $\rho=\sum_{n=1,k=1}^{d+1,d}p_{nk}\Pi_{nk}-I$, where $\Pi_{nk}$ is the projection operator onto the eigenspace of the eigenvector $\ket{k}$ of the observable $A_{n}$, and $p_{nk}$ (=tr($\Pi_{nk}\rho)$) is the probability of obtaining the corresponding eigenvalue through measurement. The relation between the elements of two MUBs can be then rewritten as $\tr(\Pi_{nk}\Pi_{ml})=\frac{1+(d\delta_{kl}-1)\delta_{nm}}{d}$. The set of probability distributions $\{p_{nk}, n=1,2,\ldots,N; k=1,2,\ldots,d\}$ of $N$ MUBs obeys the algebraic relation,
\begin{equation}\label{eq01}
    \sum_{n=1,k=1}^{N,d}p_{nk}^{2}\leq \tr(\rho^{2})+1,
\end{equation}
which was obtained in Refs.\cite{Ivanovic1992,Wu2009} independently. $C_{n}:= \sum_{k=1}^{d}p_{nk}^{2}$ is called the purity of the observable $A_{n}$. Hence, the inequality in Eq.(\ref{eq01}) is a restriction on the summation of the purities of $N$ mutually exclusive observables, and the equality is achieved when $N$ is $d+1$. When the summation of entropies of $N$ observables is maximized, this restriction has to be taken into account. The optimization of this restriction on purities was used in Ref.\cite{Puchala2015} in order to obtain lower and upper bounds of entropic uncertainty relation of N observables for pure states.
Optimal entropic certainty relation for $N$ MUBs can be obtained if, additional to the summation of probability to unity, the inequality (\ref{eq01}) is considered in the maximization of the entropy-summation of the observables. In Refs.\cite{Sanchez-Ruiz1993,Sanchez-Ruiz1995}, author found entropic certainty relation for $d+1$ MUBs, with the aid of the assumption that the purities of the observables corresponding to MUBs are constant independently. We first extend the equality in Eq.(\ref{eq01}) to $N$ MUBs for some density matrix, and then, take it as a condition on the probability distributions; thus, in turn, the purities of the observables are considered dependent on each other.  The intuitive reason behind our consideration can be seen from the following scenario. If one assumes the probability distribution of an observable as $\{p_{n1} = 1, p_{n2} = p_{n3}=\cdots = p_{nd} = 0\}$, then the probability distributions of the rest observables become equally likely as $\{p_{s1} = p_{s2} = \cdots = p_{sd} =1/d ; s = 1,2,\ldots,n-1,n+1,\ldots,N\}$, which implies that the purities of the observables corresponding to $N$ MUBs are dependent on each other. 
\begin{proposition}
Let $\{\ket{nk}, k=1,2,\ldots,d\}$ be the orthonormal basis of the observable $A_{n}$ in Hilbert space $H^{d}$. Then, for the density matrices $\rho= \sum _{n=1,k=1}^{N,d}\lambda_{nk}\dyad{nk}{nk} $, the summation of the purities of N observables is $\sum_{n=1}^{N}C_{n}:= \sum_{n=1,k=1}^{N,d}p_{nk}^{2}= \tr(\rho ^{2})+ \frac{N-1}{d}$.
\end{proposition}
\begin{proof}
When the dimension of the relevant system is a power of a prime number, the expression $\rho= \sum _{n=1,k=1}^{N,d}\lambda_{nk}\dyad{nk}{nk}$ is valid for  density matrices that can be expanded in terms of $N$ mutual unbiased bases such that $1\leq N \leq d+1$, because in this case, there are $d+1$ MUBs \cite{Wootters1989}. If the dimension is not a power of a prime number then the expression given above for the density matrices is still valid at least  when $1\leq N \leq 3$ since we know that there exist at least three MUBs in any finite dimensional Hilbert space \cite{Klappenecker2003}. 
\par Let us assume that $\rho= \sum _{n=1,k=1}^{N,d}\lambda_{nk}\ket{nk}\bra{nk}.$ 
Since $\tr(\rho)=1$ then $\sum_{n=1, k=1}^{N,d} \lambda_{nk}=1$. Furthermore, the trace of the square of density matrix leads to 
\begin{equation}
\begin{split}
 \tr(\rho^{2})&= \sum_{m,n,k,s} \lambda_{nk}\lambda_{ms}\frac{1+(d\delta_{ks}-1)\delta_{nm}}{d}\\
 &=\frac{1}{d}+\sum_{nk}\lambda_{nk}^{2}-\frac{1}{d}\sum_{n,k,s}\lambda_{nk}\lambda_{ns},     
\end{split}
\end{equation} 
and the probabilities are 
\begin{equation}
\begin{split}
 p_{nk}:&= \tr(\Pi_{nk}\rho)\\
 & = \lambda_{nk}+\frac{1}{d}\sum_{m,l}\lambda_{ml}-\frac{1}{d}\sum_{l}\lambda_{nl}.
\end{split}
\end{equation} 
If we consider the probabilities $\{p_{nk}\}$ and the coefficients $\{\lambda_{nk}\}$ as two column vectors  $\vb{p}=(p_{11},p_{12},\ldots,p_{(N)d})^{T}$ and $\vb*{\lambda}=(\lambda_{11},\lambda_{12},...,\lambda_{(N)d})^{T}$, then the relation between them can be written by means of an $Nd \times Nd$ symmetric matrix $\vb{T}$ as $\vb{p}=\vb{T}\vb*{\lambda}$. More explicitly,
\begin{equation}
   \begin{pmatrix}
    \vdots\\
    p_{nk}\\
    \vdots
    \end{pmatrix}
    = \begin{pmatrix}
    I_{d}& D_{1}&D_{2} & \ldots & D_{N-1}\\
    D_{1}& I_{d}&D_{2} & \ldots & D_{N-1}\\
     \vdots & & &\ddots  & \\
    
    D_{1} & D_{2}& \ldots &D_{N-1}& I_{d} 
    
    \end{pmatrix}
    \begin{pmatrix}
    \vdots\\

    \lambda_{nk}\\
    \vdots\\
    \end{pmatrix},
\end{equation}
where $I_{d}$ is $d \times d$ identity matrix and the matrices  $\{D_{i}\}_{i=1}^{N-1}$ are also $d \times d$ matrices such that their all entries are  $\frac{1}{d}$, that is 
\begin{equation}
    D_{1}=\cdots = D_{N-1}=D_{d}=\frac{1}{d}\begin{pmatrix}
    1 & 1& \ldots & 1\\
    1 & 1& \ldots & 1\\
     &  & \vdots &\\
    1 & 1& \ldots& 1
    \end{pmatrix}.
\end{equation}
It is easily seen that $D_{d}^2 = D_{d}$. The matrix $\vb{T}$ is not invertible which implies that a particular distribution $\vb{p}=(p_{11},p_{12},...,p_{(N)d})^{T}$ is not uniquely determined by the  density matrix. 
The summation of the purities of $N$ complementary observables is equal to the square of the norm of $\vb{p}$, $\sum_{n=1}^{N}C_{n} = \sum_{n=1,k=1}^{N,d}p_{nk}^{2}=\vb{p}^T \vb{p}$, where it reads
\begin{equation}
\begin{split}
 \vb{p}^T \vb{p} &= \vb*{\lambda}^{T} \vb{T}^{2} \vb*{\lambda}\\
 & =\frac{N}{d}+\sum_{n,k}\lambda_{nk}^{2}-\frac{1}{d}\sum_{n,k,s}\lambda_{nk}\lambda_{ns}\\
 &=\frac{N}{d} + \tr(\rho^{2})-\frac{1}{d}=\tr(\rho^{2})+\frac{N-1}{d}.
\end{split}
\end{equation}
\end{proof}

\par Consequently, this equality of purities proves the aforementioned intuitive reasoning of the fact that purities of the observables are dependent on each other. Therefore, the equality has to be taken into account when maximized the summation of the entropies. We obtain the optimal entropic certainty relation for $N$ MUBs  under  the following  conditions satisfied the probability distributions of the associated observables
\begin{align}
	\sum_{k=1}^{d}p_{nk}&=1\\
	 \sum_{n=1,k=1}^{N,d}p_{nk}^{2}&= \tr(\rho^{2})+ \frac{N-1}{d} \label{eq02}, 
\end{align} 
and under the assumption that \textit{the density matrix $\rho$ can be expressed in terms of $N$ MUBs under consideration. For $N=d+1$, the density matrix is the general one and, in turn, our following results become true for any density matrix}. Henceforth, we will abbreviate the trace of the square of density matrix as  $\Pi:=\tr\left(\rho^{2}\right)$. Our method is based on the variation of the function
\begin{equation*}
	S[\{A_{n}\}]:= \sum_{n=1}^{N}H(A_{n})=-\sum_{n=1,k=1}^{N,d}p_{nk}\ln p_{nk},
\end{equation*} 
where $H(A_{n})$ is Shannon entropy of the observable $A_{n}$. Maximization of the function $S[\{A_{n}\}]$ under the conditions given above is equivalent to the maximization of the following function
\begin{equation}
\begin{split}
&\Omega(\{p_{nk}\}):= -\sum_{n=1,k=1}^{N,d}p_{nk}\ln p_{nk}\\
&-\lambda(\sum_{n=1,k=1}^{N,d}p_{nk}^{2}-\Pi-\frac{N-1}{d})
-\beta(\sum_{k=1}^{d}p_{nk}-1),
\end{split}
\end{equation}
where $\lambda$ and $\beta$ are Lagrange multipliers. Variation of $\Omega$-function reads
\begin{equation*}
\begin{split}
\delta \Omega&=\\
 &\sum_{k=1}^{d}\left(-\sum_{n=1}^{N}\ln p_{nk}-2\lambda\sum_{n=1}^{N} p_{nk}-(\beta + N)\right)\delta p_{nk}=0,
 \end{split}
\end{equation*}
 so that the following equality must be satisfied for all $p_{nk}$'s, where none of them can be zero,
 \begin{equation}\label{eq03}
 	\sum_{n=1}^{N}\ln p_{nk}+2\lambda\sum_{n=1}^{N} p_{nk}+(\beta +N)=0, \quad k=1,2,\ldots,d.
 \end{equation}
 Without losing generality, we choose the probabilities set $\{p_{nd}=b_{n}, p_{nk}=t_{nk}b_{n}, k=1,2,\ldots,d-1; n=1,2,\ldots,N\} $. Substituting these probabilities into Eq.(\ref{eq03}), we obtain two equations
 \begin{align}
 &\sum_{n=1}^{N}\ln b_{n}+2\lambda\sum_{n=1}^{N} b_{n} = -(\beta +N)\quad for \quad k=d, \label{eq04}\\
 &\sum_{n=1}^{N}\ln t_{nk}+\sum_{n=1}^{N}\ln b_{n}+2\lambda\sum_{n=1}^{N} t_{nk}b_{n} = -(\beta +N)\quad for\quad k=1,2,\ldots,d-1\label{eq05}.
 \end{align} 
 Substituting $-(\beta+N)$ of Eq.(\ref{eq04}) into Eq.(\ref{eq05}), we obtain the following equality
 \begin{equation}
 \frac{\sum_{n=1}^{N}\ln t_{nk}}{\sum_{n=1}^{N}\left(t_{nk}-1\right)b_{n}}=-2\lambda ; \quad k=1,2,\ldots,d-1 \label{eq06}.
 \end{equation}
 The right hand side of Eq.(\ref{eq06}) is a constant number for every $k=1,2,\ldots,d-1$, so that the parameter  $t_{nk}$ must be independent of index-$k$, that is, $t_{n1}=t_{n2}= \cdots =t_{n\left(d-1\right)}=\frac{1-b_{n}}{(d-1)b_{n}}$. Consequently, we obtain the probability distributions as $\{p_{nd}=b_{n}, p_{nk} =\frac{1-b_{n}}{d-1}, k=1,2,\ldots,d-1; n=1,2,\ldots,N\}$. According to these distributions, the summation of the entropies is
 \begin{equation}\label{eq07}
 \begin{split}
  H_T(\{b_{n}\}):&= S[\{A_{n}\}] =-\sum_{n=1}^{N}b_{n}\ln b_{n}\\
  &-\sum_{n=1}^{N}\left(1-b_{n}\right)\ln\left(\frac{1-b_{n}}{d-1}\right)
 \end{split}
 \end{equation}
with the condition
\begin{equation}\label{eq08}
	\sum_{n=1}^{N}\left(d b_{n}^{2}-2b_{n}\right)=\frac{(d-1)\left[d(\Pi+1)-(d+1)\right]-N}{d},
\end{equation}
which is the revision of the condition in Eq.(\ref{eq02}), since we could not eliminate this condition at the end of the maximization of the function $S$. To find the extremum values of the function $H_{T}$, we define similarly another function as
\begin{equation}
\begin{split}
	&\Psi(\{b_{n}\}):=-\sum_{n=1}^{N}b_{n}\ln b_{n}-\sum_{n=1}^{N}\left(1-b_{n}\right)\ln\left(\frac{1-b_{n}}{d-1}\right)\\
	&-\mu\left(\sum_{n=1}^{N}\left(db_{n}^{2}-2b_{n}\right)-\frac{(d-1)\left[d(\Pi+1)-(d+1)\right]-N}{d}\right).
\end{split}
\end{equation}
The variation of $\Psi$ function reads
\begin{equation}
	\sum_{n=1}^{N}\left(\ln\left(\frac{1-b_{n}}{(d-1)b_{n}}\right)-2\mu(db_{n}-1)\right)\delta b_{n}=0 
\end{equation}
Since the infinitesimals $\{\delta b_{n}\}$ are arbitrary, the coefficients  must be zero
\begin{equation}\label{eq09}
\begin{split}
	&\ln\left(\frac{1-b_{n}}{(d-1)b_{n}}\right)-2\mu(db_{n}-1)=0 \\
	&\Rightarrow \frac{\ln\left(\frac{1-b_{n}}{(d-1)b_{n}}\right)}{db_{n}-1}=2\mu; n=1,2,\ldots,N.
\end{split}
\end{equation}
The left hand side of Eq.(\ref{eq09}) is constant, so that the parameters $b_{n}$ must be independent of index-$n$, that is, $b_{1}= b_{2} = \cdots = b_{N}$. Bearing in mind this fact, we obtain $b_{n}$ from Eq.(\ref{eq08}) as
\begin{equation}\label{eq10}
\begin{split}
	&b_{n}^{\pm}=\frac{\sqrt{N} \pm \sqrt{(d-1)\left[d(\Pi+1)-(d+1)\right]}}{d\sqrt{N}}; \\
	& \ceil*{\frac{d+1}{\Pi+1}} \leq d,
\end{split}
\end{equation}  
where $\ceil*{.}$ is ceiling function. The condition on the dimension  $d$ in Eq.(\ref{eq10}) comes from the fact that the term $\sqrt{(d-1)\left[d(\Pi+1)-(d+1)\right]}$ must be non-negative real number. The value $b_{n}^{+}$ gives the optimal upper bound of the total entropy $H_{T}$. Making the abbreviation $\alpha := \sqrt{(d-1)\left[d(\Pi+1)-(d+1)\right]}$, we obtain the optimal entropic certainty relation for $N$ MUBs as
\begin{equation}\label{eq11}
\begin{split}
	&H_{T} \leq H_{T}^{+}=N\ln\left(\frac{d(d-1)\sqrt{N}}{(d-1)\sqrt{N}-\alpha}\right)\\
	&-\frac{N+\sqrt{N}\alpha}{d}\ln\left(\frac{(d-1)(\sqrt{N}+\alpha)}{(d-1)\sqrt{N}-\alpha}\right).
\end{split}
\end{equation}
In order $b_{n}^{-}$ to be a positive real number, it requires that
\begin{equation}\label{eq12}
\begin{split}
	&b_{n}^{-}=\frac{\sqrt{N} - \sqrt{(d-1)\left[d(\Pi+1)-(d+1)\right]}}{d\sqrt{N}} >0 \\
	&\Rightarrow d< \frac{d+1}{\Pi +1}+\frac{N}{(d-1)(\Pi +1)}\leq  \frac{d+1}{\Pi +1}\left(1+\frac{1}{d-1}\right)\\
	&\Rightarrow d< \frac{d+1}{\Pi +1}+1 \rightarrow d \leq \ceil*{\frac{d+1}{\Pi+1}}.
	\end{split}
\end{equation} 
Both restrictions in Eq.(\ref{eq10}) and Eq.(\ref{eq12}) on the dimension  $d$ give a unique value, $d=\ceil*{\frac{d+1}{\Pi+1}}$. This is possible only if $ \Pi = 1/d$, which corresponds to pure mixed state $\rho = \frac{1}{d}I$. This means that $b_{n}^{-}$ cannot be a stationary value for the function $S[\{A_{n}\}]$ but an extremum \cite{Lanczos}, which is the special value $1/d$ of $b_{n}^{+}$. 

\section{Results and Discussion}

\par Our result in Eq.(\ref{eq11}) is different from the one given in Ref.\cite{Sanchez-Ruiz1995} since we have considered Eq.(\ref{eq02}) when maximized the function $S[\{A_{n}\}]$ which is satisfied by the purities of the observables $\{A_{n}\}_{n=1}^{N}$, and makes them dependent on each other. In addition, contrast to the certainty relations given in Refs\cite{Sanchez-Ruiz1995,Wu2009, Puchala2015}, our result is optimal when density matrix of the system of inquiry can be written in terms of bases of $N$ observables. Entropic certainty relation of Ref.\cite{Puchala2015} is valid only for pure states and, is not optimal. As a difference from the result in Ref.\cite{Puchala2015}, our result in Eq.(\ref{eq11}) is state-independent for pure states. When $N=d+1$, our result is novel since it is  optimal upper bound for general density matrices. 
\par we have confirmed the novelty of the result by some numerical estimations. For a pure state $\rho$ in dimension $d=2$, one can estimate the maximum value of the total entropy of the spin observables (operators) $\{\sigma_{X},\sigma_{Y},\sigma_{Z}\}$ numerically as $1.547120$, which coincides with the value of the optimal upper bound $H_{T}^{+}$ given in Eq.(\ref{eq11}); for $d=3$, $N= d+1$ and a pure state $\rho$, the (maximum) value is numerically $ \approx 3.449119$, which again almost coincides with the value $3.47025$ of $H_{T}^{+}$ (for details, see Appendix A).
\par The physical significance of entropic certainty relation rises in searching mutually coherent states, which are related with the existence of MUBs. By definition, $\ket{\psi_{coh}}$ is a mutually coherent state with respect to $N$ MUBs associated with the set of observables $\{A_{n}\}_{n=1}^{N}$, iff $\{\tr(\Pi_{nk}\ket{\psi_{coh}}\bra{\psi_{coh}})=\frac{1}{d}, \forall n,k ; n=1,2,\ldots,N; k=1,2,\ldots,d \}$. As emphasized in introduction, MUBs have important applications in the fields such as quantum cryptography and quantum state tomography. Even if the existence of $3$ MUBs is known \cite{Klappenecker2003}, whether there are more than three MUBs in non-prime power dimension is still an open question.  If $\{\ket{\psi_{k}}\}_{k=1}^{d}$ are  mutually coherent states with respect to $N$ MUBs, the set of N MUBs can be extended to $N+1$ MUBs \cite{Mandayam2014}. Stating in a reverse manner, \textit{ $(i)$ if there is no a mutually coherent state $\ket{\psi_{coh}}$ with respect to $N$ MUBs, this set of  $N$ MUBs cannot be extended to $N+1$ MUBs}. It is straightforward to see that in case of the state of the system being a mutually coherent state (with respect to N MUBs in question), total entropy of N MUBs must achieve its maximum value, that is, $N\ln(d)$. We now wish to show how our result in Eq.(\ref{eq11}) covers this fact. We assume that the density matrix of the system of inquiry could be written in terms of N MUBs and the mutually coherent state $\ket{\psi_{coh}}$, that is, 
\begin{equation}\label{eq13}
    \rho= \sum_{n=1,k=1}^{N,d}\lambda_{nk}\ket{nk}\bra{nk}+r\ket{\psi_{coh}}\bra{\psi_{coh}}.
\end{equation}
The only change in our maximization procedure for total entropy happens to the parameter $\alpha$ such that $\alpha \mapsto \bar{\alpha}= \sqrt{(d-1)\left[d(\Pi+1)-(d+1)-r^{2}(d-1)\right]}$. Therefore, we need to make the revision $H_{T}^{+}(N,d,\alpha) \mapsto H_{T}^{+}(N, d,\bar{\alpha})$ in Eq.(\ref{eq11}). Now, if $\rho$ is a mutually coherent state with respect to $N$ MUBs, it must be $\forall \lambda_{nk}=0, r=1$, which makes the parameter $\bar{\alpha}=0$, and thereby, $H_{T}^{+}(N,d,0)$ reduces to $N\ln(d)$ that was to be shown. In addition to the numerical justification, this is another justification of the fact that our result in Eq.(\ref{eq11}) is indeed optimal. Since $r\rightarrow 0$ then $\bar{\alpha}\rightarrow \alpha$, and since result in Eq.(\ref{eq11}) is optimal upper bound, we can, in consequence, assert that, \textit{ $(ii)$ if the optimal upper bound in Eq.(\ref{eq11}) cannot be exceeded, there is no a mutually coherent state with respect to $N$ MUBs}. As a result, from the two premises $(i)$ and $(ii)$ above, we make the following inference: \textit{(iii) if the optimal upper bound for $N$ MUBs in Eq. (\ref{eq11}) cannot be exceeded, this set of $N$ MUBs cannot be extended to $N+1$ MUBs}. This inference sets forth a quantitative criterion for the existence of  mutually coherent states, and thus, for the extendibilty of MUBs. For instance, since the existence of $4$ MUBs in 6-dimensional Hilbert space is still conundrum, this criterion can be used as a numerical ground in order to show the non-existence of $4th$ MUB. If the upper bound in Eq.(\ref{eq11}) cannot be exceeded for $3$ MUBs in six dimensional Hilbert space, then there is no fourth MUB.
\par We now wish to give an application of entropic certainty relation in Eq.(\ref{eq11}) to the mutual information. If the total entropy of the observables set $\{A_{n}\}_{n=1}^{d+1}$ has a lower bound such as $\sum_{n}H(A_{n})\geq q$, then that summation of the observables, where each of them is conditioned with a classical memory $Y$, satisfies the inequality $\sum_{n}H(A_{n}\mid Y)\geq q$; more formally, $\sum_{n}H(A_{n})\geq q \Rightarrow \sum_{n}H(A_{n}\mid Y)\geq q$  (see, Ref.\cite[p.22]{Coles2017}). From the definition of the mutual information, we can now write
\begin{equation}
\begin{split}
    &I(A_{n}:Y):= H(A_{n})-H(A_{n}\mid Y) \\
    &\Rightarrow \sum_{n}\left( H(A_{n})-I(A_{n}:Y)\right)=\sum_{n}H(A_{n}\mid Y)\\
    & \Rightarrow \sum_{n}\left(H(A_{n})-I(A_{n}:Y)\right)\geq q.
\end{split}
\end{equation}
For the complementary observables $\{A_{n}\}_{n=1}^{d+1}$ in $d$-dimensional Hilbert space, $q= (d+1)\ln \left(\frac{d+1}{\Pi+1}\right)$ \cite{Rastegin2013} and using the inequality in Eq.(\ref{eq11}) for $d+1$ MUBs, we obtain an upper bound on the summation of the mutual information as
\begin{equation}\label{eq14}
    \begin{split}
        \sum_{n=1}^{d+1}I(A_{n}:Y) & \leq (d+1)\ln\left(\frac{d(d-1)(\Pi+1)\sqrt{d+1}}{(d+1)\left[(d+1)\sqrt{d+1}-\alpha\right]}\right)\\
        & - \frac{d+1+\sqrt{d+1}\alpha}{d}\ln\left(\frac{(d-1)\sqrt{d+1}+\alpha}{(d-1)\sqrt{d+1}-\alpha}\right).
    \end{split}
\end{equation}
\section{Conclusion}
We have obtained the optimal upper bound of the entropic uncertainty relation  for  $N$ MUBs if the density matrix of the relevant system can be expressed in terms of $N$ MUBs. This bound implies that the entropies of the observables cannot achieve to their maximum values ($\ln(d)$) simultaneously. The crucial point in our derivation is the condition satisfied by the purities of the observables given in Eq.(\ref{eq01}). As pointed out, the purities of the observables corresponding to $N$ MUBs are dependent on each other; therefore, we have considered the equality in  Eq.(\ref{eq02}) in the maximization of the total entropy. If an equality relation for the summation of the purities of $N$ MUBs exists for a general density matrix, our result can be extended directly. An equality of this sort will be related with the dimension of the system ($d$), the density matrix ($\rho$) and the number of MUBs ($N$). As another choice, if a way of how to take the inequality in Eq.(\ref{eq01}) into account can be found, an optimal entropic uncertainty relation for $N$ observables can be again obtained for a general density matrix. Eq.(\ref{eq01}) is a non-holonomic condition on the summation of the entropies. It seems that the maximization  under this non-holonomic condition cannot be solved by the method given in Ref.\cite{Dreyfus1963}, which is about the variational calculation of a (at least piece wise) continuous function with inequality constraints.
\par We have shown that our result in Eq.(\ref{eq11}) provides a criterion for the existence of mutually coherent states, which are related with the existence of MUBs. Two questions can be argued in connection with the criterion: Can we assert that \textit{if there is no a mutually coherent states, the optimal upper bound in Eq.(\ref{eq11}) cannot be exceeded}? The second question is that: Can a new MUB be constructed, starting from a mutually coherent state (with respect to the old $N$ MUBs) that we find? To answer the first question, one needs a detailed logical analysis of the premises $(i)$ and $(ii)$ given above. As for the second question, we would like to just orient the attentions to two related works \cite{Mandayam2014} and \cite{Szanto2016} for now. 
\par We have also applied entropic certainty relation to the summation of the mutual information of $d+1$ complementary observables conditioned with a classical memory; one can make use of Eq.(\ref{eq14}) to detect whether the observables are correlated. In a scenario of detecting this correlation between spin observables $\{\sigma_{X},\sigma_{Y},\sigma_{Z}\}$, the optimal lower bound for entropic uncertainty relation is $q=\ln{4}-\{\frac{1+\abs{\vb{r}}}{2}\ln\frac{1+\abs{\vb{r}}}{2}+\frac{1-\abs{\vb{r}}}{2}\ln\frac{1-\abs{\vb{r}}}{2}\}$ \cite{Sanchez-Ruiz1993}, where $\vb{r}$ is Bloch vector in the density matrix representation, $\rho=\frac{1}{2}\left(I+\vb{r}.\vb*{\sigma}\right)$ with $\vb*{\sigma}=\left(\sigma_{X},\sigma_{Y},\sigma_{Z}\right)^{T} $.  The inequality given in Eq.(\ref{eq13}) can be revised depending on the lower bound $q$ of the summation of the entropies.

\appendix
\section{Appendix A: Probability distributions of d+1 MUBs in dimensions d=2 and d=3 for a pure state}
\subsection{Probability distributions in dimensions d=2}
In dimension d=2, a pure density matrix in computational basis $\{\ket{0},\ket{1}\}$ is
\begin{equation*}
\textbf{$\rho$}=
\begin{bmatrix}
\mid\alpha\mid ^{2} \quad \alpha\beta^{*}\\
\alpha^{*} \beta \quad \mid\beta\mid ^{2}
\end{bmatrix}.
\end{equation*}
In addition, taking the eigenstates of spin operators as columns for constructing the unitary matrices
\begin{equation*}
	U_{z}=\begin{bmatrix}
	1 & 0\\
	0 & 1
	\end{bmatrix}, U_{x}=\frac{1}{\sqrt{2}} \begin{bmatrix}
	1 & 1\\
	1 & -1
	\end{bmatrix}, U_{y}=\frac{1}{\sqrt{2}} \begin{bmatrix}
	1 & 1\\
	i & -i
	\end{bmatrix},
\end{equation*}
we can calculate the probabilities as $ p_{nk}= \bra{1k}U_{n}^{\dagger}\rho U_{n}\ket{1k}$, where $\{\ket{11}=\ket{0}, \ket{12}=\ket{1}\}$ is computational basis. Without losing generality, if we choose $\alpha =\sqrt{r}$ and $\beta=\sqrt{1-r}\exp(i\phi)$, then we obtain the probability distributions of spin observables $S_{z}, S_{x}, S_{y}$ as in the Table \ref{table:1}.
\begin{table}[H]
	\centering
	\begin{tabular}{ |l|l|l| }
		\hline
		\multicolumn{3}{|c|}{Table of MUBs and their probabilities, d=2} \\
		\hline
		$S_{z}$ & $p_{11}= r$ & $p_{12}=1-r$ \\
		\hline
		$S_{x} $ & $p_{21}=\frac{1}{2}(1+2\sqrt{r(1-r)}cos(\phi))$ & $p_{22}=\frac{1}{2}(1-2\sqrt{r(1-r)}cos(\phi))$ \\
		
		\hline 
		$S_{y}$ & $p_{31}=\frac{1}{2}(1-2\sqrt{r(1-r)}sin(\phi))$ & $p_{32}=\frac{1}{2}(1+2\sqrt{r(1-r)}sin(\phi))$\\  
		\hline
	\end{tabular}
	\caption{The probability distributions table of MUBs in d=2 when the density matrix is a pure state. The first column on left stands for MUBs ($S_{n}, n= z,x,y.$), and the others for probabilities of obtaining their first and second eigenvalues respectively.} 
	\label{table:1}
\end{table}
Writing total Shannon entropy of the observables ($S_{n}, n= z,x,y.$)
\begin{equation*}
    S_{T}(r,\phi):= -\sum_{n=1,k=1}^{3,2} p_{nk} \ln(p_{nk}),
\end{equation*}
We can estimate numerically the maximum value of $S_{T}$ by adjusting the parameters $r$ and $\phi$. The maximum values is $1.547120$, achieving when $r= \frac{1}{2}$ and $\phi= \frac{\pi}{4}$. 
\subsection{Probability distributions in dimension d=3}
Like in dimension d=2, the general pure density matrix in dimension d=3 can be written as the follows
\begin{equation*}
\textbf{$\rho$}=
\begin{bmatrix}
\mid\alpha\mid ^{2} & \alpha\beta^{*} & \alpha\gamma ^{*}\\
\alpha^{*} \beta & \mid\beta\mid ^{2} & \beta \gamma ^{*}\\
\alpha^{*} \gamma & \beta^{*}\gamma & \mid \gamma \mid ^{2}\\
\end{bmatrix},
\end{equation*}
and the unitary matrices are
\begin{equation*}
	U_{1}=\begin{bmatrix}
	1 & 0 & 0\\
	0 & 1 & 0\\
	0 & 0 & 1
	\end{bmatrix}, U_{2}=\frac{1}{\sqrt{3}}\begin{bmatrix}
	1 & 1 & 1\\
	1 & \omega & \omega ^{2}\\
	1 & \omega ^{2} & \omega
	\end{bmatrix}, U_{3}=\frac{1}{\sqrt{3}}\begin{bmatrix}
	1 & 1 & 1\\
	\omega & \omega^{2} & 1\\
	\omega & 1 & \omega^{2}
	\end{bmatrix}, U_{4}=\frac{1}{\sqrt{3}}\begin{bmatrix}
	1 & 1 & 1\\
	\omega^{2} & \omega & 1\\
	\omega^{2} & 1 & \omega
	\end{bmatrix}, 
\end{equation*}
 where $\omega=\exp(\frac{2\pi i}{3})$. Then, the probability of obtaining the eigenvalue $\lambda_{k}$ of the observable $A_{n}$ is $p_{nk}=\bra{1k}[U_{n}^{\dagger}\rho U_{n}\ket{1k}$. Without losing generality, we choose $\alpha =\sqrt{r}, \beta=\sqrt{q}\exp(i \phi_{1})$ and $\gamma =\sqrt{1-(r+q)}\exp(i \phi_{2})$, leading to the probability distributions in Table \ref{table:2}:
 \begin{table}[H]
 	\centering
 	\begin{tabular}{ |l|l|l|l| }
 		\hline
 		\multicolumn{4}{|c|}{Table of MUBs and their probabilities, d=3 } \\
 		\hline
 		$A_{1}$ & $p_{11}= r$ & $p_{12}=q$ &$p_{13}= 1-(r+q)$ \\
 		\hline
 		$A_{2} $ & $p_{21}=\frac{1}{3}(1+2f_{21})$ & $p_{22}=\frac{1}{3}(1+2f_{22})$ & $p_{23}=\frac{1}{3}(1+2f_{23})$\\
 		
 		\hline 
 		$A_{3}$ & $p_{31}=\frac{1}{3}(1+2f_{31})$ & $p_{32}=\frac{1}{3}(1+2f_{32})$ & $p_{33}=\frac{1}{3}(1+2f_{33})$\\  
 		\hline
 		$A_{4}$ & $p_{41}= \frac{1}{3}(1+2f_{41})$ & $p_{42}=\frac{1}{3}(1+2f_{42})$ &$ p_{43}=\frac{1}{3}(1+2f_{43})$ \\
 		\hline
 	\end{tabular}
 	\caption{The probability distributions table of MUBs in d=3 when the density matrix is pure state. The first column on left stands for MUBs ($A_{n}, n=1,2,3,4.$), and the others for probabilities of obtaining their first, second and third eigenvalues respectively.}
 	\label{table:2}
 \end{table}
The functions $f_{nk}$'s are as follows
\begin{align*}
	f_{21}&=\sqrt{rq} cos(\phi_{1})+\sqrt{r(1-(r+q))}cos(\phi_{2})+\sqrt{q(1-(r+q))}cos(\phi_{1}-\phi_{2})\\
	f_{22}&=\sqrt{rq} cos(\phi_{1}-2\pi/3)+\sqrt{r(1-(r+q))}cos(\phi_{2}-4\pi/3)+\sqrt{q(1-(r+q))}cos(\phi_{1}-\phi_{2}+2 \pi/3)\\
	f_{23}&=\sqrt{rq} cos(\phi_{1}-4\pi/3)+\sqrt{r(1-(r+q))}cos(\phi_{2}-2\pi/3)+\sqrt{q(1-(r+q))}cos(\phi_{1}-\phi_{2}+4\pi/3)\\
\end{align*}
\begin{align*}
f_{31}&=\sqrt{rq} cos(\phi_{1}-2\pi/3)+\sqrt{r(1-(r+q))}cos(\phi_{2}-2\pi/3)+\sqrt{q(1-(r+q))}cos(\phi_{1}-\phi_{2})\\
f_{32}&=\sqrt{rq} cos(\phi_{1}-4\pi/3)+\sqrt{r(1-(r+q))}cos(\phi_{2})+\sqrt{q(1-(r+q))}cos(\phi_{1}-\phi_{2}+2\pi/3)\\
f_{33}&=\sqrt{rq} cos(\phi_{1})+\sqrt{r(1-(r+q))}cos(\phi_{2}-4\pi/3)+\sqrt{q(1-(r+q))}cos(\phi_{1}-\phi_{2}+4\pi/3)\\	
\end{align*}	
\begin{align*}	
	f_{41}&=\sqrt{rq} cos(\phi_{1}-4\pi/3)+\sqrt{r(1-(r+q))}cos(\phi_{2}-4\pi/3)+\sqrt{q(1-(r+q))}cos(\phi_{1}-\phi_{2})\\
	f_{42}&=\sqrt{rq} cos(\phi_{1}-2\pi/3)+\sqrt{r(1-(r+q))}cos(\phi_{2})+\sqrt{q(1-(r+q))}cos(\phi_{1}-\phi_{2}+4\pi/3)\\
	f_{43}&=\sqrt{rq} cos(\phi_{1})+\sqrt{r(1-(r+q))}cos(\phi_{2}-2\pi/3)+\sqrt{q(1-(r+q))}cos(\phi_{1}-\phi_{2}+2\pi/3)\\
\end{align*}
Like in d=2, the maximum value of total Shannon entropy  $S_{T}(r,q,\phi_{1}, \phi_{2})$ can be estimated, searching over its parameters $r,q,\phi_{1}$ and $\phi_{2}$. We obtained numerically the (maximum) value as $\approx 3.44911877719$, achieving when $r=0.21, q=0.395, \phi_{1}=\phi_{2}=5.236 \approx \frac{5\pi}{3}$. Since the computer we used is not powerful enough, we made the search over two variable while taking the others constant. However, a precise search must be performed by varying the whole parameters simultaneously. The theoretical value ( the value of $H_{T}^{+}$) can be numerically achieved if a more powerful computer is used.    
%%%%%%%%%  END %%%%%%%%%%%%%
%%%% Bibliography %%%%%%%%%%%%%

\end{document}